\documentclass[12pt]{article}

\usepackage{amssymb}
\usepackage{fullpage}
\usepackage{graphicx}

\newtheorem{theorem}{Theorem}
\newtheorem{lemma}{Lemma}

\newcommand{\IR}{\mathbb{R}}
\newcommand{\CTO}{\textsf{\sc CompareToOptimal}}
\newcommand{\MST}{\mathord{\it MST}}
\newcommand{\Yao}{\mathord{\it \Upsilon}_6(P,S)}
\newcommand{\keep}{\mathord{\it keep}}
\newcommand{\TRUE}{\mathord{\it true}}
\newcommand{\FALSE}{\mathord{\it false}}
\newcommand{\OP}{\mathord{\it output}}

\newcommand{\qed}{\rule{0.5em}{1.5ex}}
\newcommand{\fqed}{{\hfill~\qed}}
\newenvironment{proof}{{\noindent \bf Proof.}}
                      {{\hfill \fqed} \vspace{1em}}

\title{An Optimal Algorithm for the Euclidean Bottleneck Full Steiner 
       Tree Problem}
\author{
Ahmad Biniaz\thanks{School of Computer Science, Carleton University, 
                    Ottawa, Canada. Research supported by NSERC.}
\and 
Anil Maheshwari\footnotemark[1]
\and 
Michiel Smid\footnotemark[1]
}
\date{\today}

\begin{document} 

\maketitle 

\begin{abstract} 
Let $P$ and $S$ be two disjoint sets of $n$ and $m$ points in the plane, 
respectively. We consider the problem of computing a Steiner tree whose 
Steiner vertices belong to $S$, in which each point of $P$ is a leaf, 
and whose longest edge length is minimum. We present an algorithm that 
computes such a tree in $O((n+m)\log m)$ time, improving the previously 
best result by a logarithmic factor. We also prove a matching lower 
bound in the algebraic computation tree model. 
\end{abstract} 

\section{Introduction} 
Let $P$ and $S$ be two disjoint sets of $n$ and $m$ points in the plane, 
respectively. A \emph{full Steiner tree} of $P$ \emph{with respect to} 
$S$ is a tree $\mathcal{T}$ with vertex set $P \cup S'$, for some 
non-empty subset $S'$ of $S$, in which each point of $P$ is a leaf. 
Such a tree $\mathcal{T}$ consists of a \emph{skeleton tree}, which is 
the part of $\mathcal{T}$ that spans $S'$, and \emph{external edges}, 
which are the edges of $\mathcal{T}$ that are incident on the points 
of $P$.  

The \emph{bottleneck length} of a full Steiner tree is defined to be 
the Euclidean length of a longest edge. 
An \emph{optimal bottleneck full Steiner tree} is a full Steiner tree 
whose bottleneck length is minimum. 
In~\cite{aa-ebfst-11}, Abu-Affash shows that such an optimal tree can 
be computed in $O((n+m) \log^2 m)$ time. In this paper, we improve the 
running time by a logarithmic factor and prove a matching lower bound. 
That is, we prove the following result:  

\begin{theorem}  \label{thm1}  
Let $P$ and $S$ be disjoint sets of $n$ and $m$ points in the plane, 
respectively. An optimal bottleneck full Steiner tree of $P$ with 
respect to $S$ can be computed in $O((n+m) \log m)$ time, which is 
optimal in the algebraic computation tree model. 
\end{theorem} 

\section{The algorithm} 

\subsection{Preprocessing} 
We compute a Euclidean minimum spanning tree $\MST(S)$ of the point 
set $S$, which can be done in $O(m \log m)$ time. Then we compute the 
bipartite graph $\Yao$ with vertex set $P \cup S$ that is defined as 
follows: Consider a collection of six cones, each of angle $\pi/3$ and 
having its apex at the origin, that cover the plane. For each point 
$p$ of $P$, translate these cones such that their apices are at $p$. 
For each of these translated cones $C$ for which 
$C \cap S \neq \emptyset$, the graph $\Yao$ contains one edge 
connecting $p$ to a nearest neighbor in $C \cap S$. (This is a variant 
of the well-known Yao-graph as introduced in~\cite{y-cmstk-82}.) 
Using an algorithm of Chang \emph{et al.}~\cite{cht-oacov-90}, 
together with a point-location data structure, the graph $\Yao$ can 
be constructed in $O( (n+m) \log m)$ time.  

The entire preprocessing algorithm takes $O( (n+m) \log m)$ time.  

\subsection{A decision algorithm} 
Let $\lambda^*$ denote the \emph{optimal bottleneck length}, i.e., 
the bottleneck length of an optimal bottleneck full Steiner tree of $P$ 
with respect to $S$. 

In this section, we present an algorithm that decides, for any given 
real number $\lambda>0$, whether $\lambda^* < \lambda$ or 
$\lambda^* \geq \lambda$. This algorithm starts by removing from 
$\MST(S)$ all edges having length at least $\lambda$, resulting in 
a collection $T_1,T_2,\ldots$ of trees. The algorithm then computes 
the set $J$ of all indices $j$ for which the following holds: 
Each point $p$ of $P$ is connected by an edge of $\Yao$ to some point 
$s$, such that (i) $s$ is a vertex of $T_j$ and (ii) the Euclidean 
distance $|ps|$ is less than $\lambda$. As we will prove later, 
this set $J$ has the property that it is non-empty if and only if 
$\lambda^* < \lambda$. The formal algorithm is given in 
Figure~\ref{figalg}. 

\begin{figure}
\begin{center}
\begin{quote}
\begin{tabbing} 
{\bf Algorithm} $\CTO(\lambda)$;  \\ 
remove from $\MST(S)$ all edges having length at least $\lambda$; \\
denote the resulting trees by $T_1,T_2,\ldots$; \\ 
number the points of $P$ arbitrarily as $p_1,p_2,\ldots,p_n$; \\ 
$J := \emptyset$; \\ 
{\bf for each} edge $(p_1,s)$ in $\Yao$  \\
{\bf do} \= $j :=$ index such that $s$ is a vertex of $T_j$; \\ 
         \> {\bf if} $|p_1 s| < \lambda$ \\
         \> {\bf then} $J := J \cup \{ j \}$ \\
         \> {\bf endif} \\
{\bf endfor}; \\
{\bf for} $i:=2$ {\bf to} $n$ \\ 
{\bf do} \= {\bf for each} $j \in J$ \\ 
         \> {\bf do} $\keep(j) := \FALSE$ \\ 
         \> {\bf endfor}; \\ 
         \> {\bf for each} edge $(p_i,s)$ in $\Yao$  \\  
         \> {\bf do} \= $j :=$ index such that $s$ is a vertex of 
                        $T_j$; \\ 
         \>          \> {\bf if} $j \in J$ and $|p_i s| < \lambda$ \\ 
         \>          \> {\bf then} $\keep(j) := \TRUE$ \\ 
         \>          \> {\bf endif} \\ 
         \> {\bf endfor}; \\  
         \> $J := \{ j \in J : \keep(j) = \TRUE \}$ \\ 
{\bf endfor}; \\ 
return the set $J$ 
\end{tabbing}
\end{quote}
\caption{This algorithm takes as input a real number $\lambda$ and 
         returns a set $J$. This set $J$ is non-empty if and only if 
         $\lambda^* < \lambda$.}
\label{figalg}
\end{center}
\end{figure}

Observe that, at any moment during the algorithm, the set $J$ has 
size at most six. Therefore, the running time of this algorithm is 
$O(n+m)$. 

Before we prove the correctness of the algorithm, we introduce the 
following notation. Let $j$ be an arbitrary element in the output set  
$J$ of algorithm $\CTO(\lambda)$. It follows from the algorithm that, 
for each $i$ with $1 \leq i \leq n$, there exists a point $s_i$ in 
$S$ such that 
\begin{itemize} 
\item $s_i$ is a vertex of $T_j$, 
\item $(p_i,s_i)$ is an edge in $\Yao$, and 
\item $|p_i s_i| < \lambda$. 
\end{itemize} 
We define $\mathcal{T}_j$ to be the full Steiner tree with skeleton tree 
$T_j$ and external edges $(p_i,s_i)$, $1 \leq i \leq n$. 
Observe that, since each edge of $T_j$ has length less than $\lambda$, 
the bottleneck length of $\mathcal{T}_j$ is less than $\lambda$. 
Therefore, we have proved the following lemma. 

\begin{lemma}  \label{lem1}  
Assume that the output $J$ of algorithm $\CTO(\lambda)$ is non-empty. 
Then $\lambda^* < \lambda$. 
\end{lemma} 

The following lemma states that the converse of Lemma~\ref{lem1} 
also holds. 

\begin{lemma}  \label{lem2}  
Assume that $\lambda^* < \lambda$. Then the output $J$ of algorithm 
$\CTO(\lambda)$ has the following two properties: 
\begin{enumerate} 
\item $J \neq \emptyset$ and 
\item $J$ contains an element $j$ such that $\mathcal{T}_j$ is a full 
      Steiner tree, whose skeleton tree $T_j$ has bottleneck length 
      less than $\lambda$, and in which each external edge has length 
      at most $\lambda^*$. 
\end{enumerate} 
\end{lemma} 
\begin{proof} 
Consider an optimal bottleneck full Steiner tree, let $T^*$ be its 
skeleton tree, and denote its external edges by $(p_i,s_i)$, 
$1 \leq i \leq n$; thus, each $s_i$ is a vertex of $T^*$. Each edge
of this optimal tree has length at most $\lambda^*$. 

We may assume that $T^*$ is a subtree of $\MST(S)$; see 
Lemma~2.1 in Abu-Affash~\cite{aa-ebfst-11}. 
Since each edge of $T^*$ has length at most $\lambda^*$, which is less 
than $\lambda$, there exists an index $j$, such that $T^*$ is a subtree 
of $T_j$. We will prove that, at the end of algorithm $\CTO(\lambda)$, 
$j$ is an element of the set $J$. 

Let $i$ be any index with $1 \leq i \leq n$. Recall that the graph 
$\Yao$ uses cones of angle $\pi/3$. Consider the cone with apex $p_i$ 
that contains $s_i$. This cone contains a point $s'_i$ of $S$ such that 
$(p_i,s'_i)$ is an edge in $\Yao$. (It may happen that $s'_i = s_i$.) 
Since $|p_i s'_i| \leq |p_i s_i|$, we have 
$|s_i s'_i| \leq |p_i s_i| \leq \lambda^* < \lambda$.      

Consider the path in $\MST(S)$ between $s_i$ and $s'_i$. It follows from 
basic properties of minimum spanning trees that each edge on this path 
has length at most $|s_i s'_i| < \lambda$. Therefore, $s'_i$ is a 
vertex of the tree $T_j$. 

It follows from algorithm $\CTO(\lambda)$ that, when $p_i$ is considered, 
the index $j$ is added to $J$ if $i=1$, and $j$ stays in $J$ if 
$i \geq 2$. Thus, at the end of the algorithm, $j$ is an element of the 
set $J$, proving the first claim in the lemma. 

The full Steiner tree $\mathcal{T}_j$, having skeleton tree $T_j$ and 
external edges $(p_i,s'_i)$ for $1 \leq i \leq n$, satisfies the second 
claim in the lemma. 
\end{proof} 

\subsection{Binary search and completing the algorithm} 
Let $k$ denote the number of distinct lengths of the edges of $\MST(S)$, 
and let $\lambda_1 < \lambda_2 < \ldots < \lambda_k$ denote the sorted 
sequence of these edge lengths. Define $\lambda_0 := 0$ and 
$\lambda_{k+1} := \infty$. Using algorithm $\CTO$ to perform a binary 
search in the sequence $\lambda_0, \lambda_1  , \ldots , \lambda_{k+1}$, 
we obtain an index $\ell$ with $1 \leq \ell \leq k+1$, such that 
$\lambda_{\ell-1} \leq \lambda^* < \lambda_{\ell}$. 

Since algorithm $\CTO$ takes $O(n+m)$ time, the total time for the 
binary search is $O((n+m) \log m)$. 

Run algorithm $\CTO(\lambda_{\ell})$. Since 
$\lambda^* < \lambda_{\ell}$, it 
follows from Lemma~\ref{lem2} that, at the end of this algorithm, the 
set $J$ contains an index $j$ such that $\mathcal{T}_j$ is a full 
Steiner tree, whose skeleton tree $T_j$ has bottleneck length less 
than $\lambda$, and in which each external edge has length at most 
$\lambda^*$. Since $T_j$ is a subtree of $\MST(S)$, it follows that 
each edge of $T_j$ has length at most $\lambda_{\ell-1}$, which is at 
most $\lambda^*$. Thus, $\mathcal{T}_j$ is a full Steiner tree with 
bottleneck length at most $\lambda^*$. By definition of $\lambda^*$, 
it then follows that the bottleneck length of $\mathcal{T}_j$ is 
equal to $\lambda^*$.  

Thus, to complete the algorithm, we run algorithm 
$\CTO(\lambda_{\ell})$ and 
consider its output $J$. For each of the at most six elements $j$ of $J$, 
we construct the full Steiner tree $\mathcal{T}_j$ and compute its 
bottleneck length $\lambda_j^*$. For any index $j$ that minimizes 
$\lambda_j^*$, $\mathcal{T}_j$ is an optimal bottleneck full Steiner 
tree. This final step completes the algorithm and takes $O(n+m)$ time.   
This proves the first part of Theorem~\ref{thm1}. 

\section{The lower bound} 
In this section, we prove that our algorithm is optimal in the 
algebraic computation tree model; refer to Ben-Or~\cite{b-lbact-83} 
for the definition of this model.

\subsection{The case when $n$ is small as compared to $m$}  
We start by assuming that $n=O(m)$. We will prove that the problem of 
computing an optimal bottleneck full Steiner tree has a lower bound of 
$\Omega(m \log m)$, which is $\Omega((n+m) \log m)$. 

Consider a sequence $s_1,s_2,\ldots,s_m$ of real numbers. The 
\emph{maximum gap} of these numbers is the largest distance between 
any two elements that are consecutive in the sorted order of this 
sequence. Lee and Wu~\cite{lw-gcslp-86} have shown that, in the 
algebraic computation tree model, computing the maximum gap takes 
$\Omega(m \log m)$ time. 

Consider the following algorithm that takes as input a sequence 
$s_1,s_2,\ldots,s_m$ of real numbers: 
\begin{enumerate} 
\item Compute the minimum and maximum elements in the input sequence, 
      compute the absolute value $\Delta$ of their difference, and 
      compute the value $g = \Delta/(m+1)$. 
\item Compute the set $S = \{ (s_i,0) : 1 \leq i \leq m \}$, 
      a set $P_1$ consisting of $n/2$ points that are to the left 
      of $(s_1,0)$ and have distance at most $g/2$ to $(s_1,0)$, 
      a point set $P_2$ consisting of $n/2$ points that are to the right 
      of $(s_m,0)$ and have distance at most $g/2$ to $(s_m,0)$. 
      Let $P$ be the union of $P_1$ and $P_2$. 
\item Compute an optimal bottleneck full Steiner tree $\mathcal{T}$ of 
      $P$ with respect to $S$, and compute the length $\lambda^*$ of a 
      longest edge in $\mathcal{T}$. 
\item Return $\lambda^*$.       
\end{enumerate}  
Let $G$ be the maximum gap of the sequence $s_1,s_2,\ldots,s_m$, and 
observe that $G \geq g$. It is not difficult to see that $G = \lambda^*$. 
Thus, the above algorithm solves the maximum gap problem and, therefore, 
takes $\Omega(m \log m)$ time. Since $n=O(m)$, the running time of this 
algorithm is $O(m+n)=O(m)$ plus the time needed to compute $\mathcal{T}$. 
It follows that the problem of computing an optimal bottleneck full 
Steiner tree has a lower bound of $\Omega(m \log m)$.

\subsection{The case when $n$ is large as compared to $m$}  
We now assume that $n=\Omega(m)$. We will prove that the problem of 
computing an optimal bottleneck full Steiner tree has a lower bound of 
$\Omega(n \log m)$, which is $\Omega((n+m) \log m)$. 

A sequence $p_1,p_2,\ldots,p_n$ of points in the plane is specified 
by $2n$ real numbers. We identify such a sequence with the point 
$(p_1,p_2,\ldots,p_n)$ in $\IR^{2n}$. For each integer $i$ with 
$1 \leq i \leq m$, let $c_i$ be the point $(i,1)$. Define the subset 
$V$ of $\IR^{2n}$ as 
\[ V = \{ (p_1,p_2,\ldots,p_n) \in \IR^{2n} : 
           \{ p_1,p_2,\ldots,p_n\} \subseteq \{c_1,c_2,\ldots,c_m\} \} .
\]
For any function $f : \{1,2,\ldots,n\} \rightarrow \{1,2,\ldots,m\}$, 
define the point $P_f = ( c_{f(1)} , c_{f(2)} , \ldots , c_{f(n)} )$. 
Since there are $m^n$ such functions $f$, we obtain $m^n$ different 
points $P_f$, each one belonging to the set $V$. The set $V$ is in fact 
equal to the set of these $m^n$ points $P_f$ and, therefore, $V$ has 
exactly $m^n$ connected components. Thus, by Ben-Or's 
theorem~\cite{b-lbact-83}, any algorithm that decides whether a given 
point $(p_1,p_2,\ldots,p_n)$ belongs to $V$ has worst-case running time 
$\Omega(n \log m)$. 

Now consider the following algorithm that takes as input a sequence 
$p_1,p_2,\ldots,p_n$ of points in the plane: 
\begin{enumerate} 
\item Compute the set $S = \{ (i,0) : 1 \leq i \leq m \}$. 
\item Let $p=(0,0)$ and $q=(m+1,0)$, and compute the set 
      $P' = \{p,q\} \cup \{p_1,p_2,\ldots,p_n\}$. 
\item Compute an optimal bottleneck full Steiner tree $\mathcal{T}$ 
      of $P'$ with respect to $S$. 
\item Set $\OP = \TRUE$. 
\item For each $j$ with $1 \leq j \leq n$, do the following: 
      \begin{enumerate} 
      \item Let $i$ be the index such that $p_j$ and $(i,0)$ are 
            connected by an external edge in $\mathcal{T}$.  
      \item If $p_j \neq c_i$, set $\OP = \FALSE$. 
      \end{enumerate} 
\item Return $\OP$. 
\end{enumerate} 

If the output of the algorithm is $\TRUE$, then each $p_j$ is equal to 
some $c_i$ and, therefore, the point $(p_1,p_2,\ldots,p_n)$ belongs to 
the set $V$. 

Assume that $(p_1,p_2,\ldots,p_n) \in V$. The (unique) optimal bottleneck 
full Steiner tree of $P'$ with respect to $S$ is the union of 
(i) the path connecting the points of $S$ sorted from left to right 
(this is the skeleton tree), 
(ii) the edge connecting $p$ with $(1,0)$ and the edge connecting 
$q$ with $(m,0)$ (these are external edges), and 
(iii) edges that connect each point $p_j$ of $P$ to the point $c_i$ 
having the same $x$-coordinate as $p_j$ (these are also external edges). 
It then follows from the algorithm that the output is $\TRUE$. 

Thus, the algorithm correctly decides whether any given point 
$(p_1,p_2,\ldots,p_n)$ belongs to $V$. By the result above, the 
worst-case running time of this algorithm is $\Omega(n \log m)$. 
Since $m=O(n)$, the running time of this algorithm is $O(m+n)=O(n)$ 
plus the time needed to compute $\mathcal{T}$. 
It follows that the problem of computing an optimal bottleneck full 
Steiner tree has a lower bound of $\Omega(n \log m)$. 

This completes the proof of the lower bound in Theorem~\ref{thm1}. 

\bibliographystyle{plain}
\bibliography{bottleneckSteiner}

\end{document}